\documentclass[12pt]{article}

\usepackage{amssymb}
\usepackage{graphicx}
\usepackage{epsfig}
\usepackage{latexsym}
\usepackage{amsmath}
\usepackage{caption}
\usepackage{enumerate}
\usepackage{appendix}
\usepackage[T1]{fontenc}
\usepackage{amsthm}

\newcommand{\bea}{\begin{eqnarray}}
\newcommand{\eea}{\end{eqnarray}}
\newcommand{\be}{\begin{equation}}
\newcommand{\ee}{\end{equation}}
\newcommand{\beann}{\begin{eqnarray*}}
\newcommand{\eeann}{\end{eqnarray*}}
\newcommand{\balnn}{\begin{align*}}
\newcommand{\ealnn}{\end{align*}}

\newcommand{\nn}{\nonumber}

\newcommand{\R}{{\mathbb R}}

\newcommand{\ba}{\begin{array}}
\newcommand{\ea}{\end{array}}
\newcommand{\bd}{\begin{displaymath}}
\newcommand{\ed}{\end{displaymath}}

\newcommand{\1}{{\mathbf 1}}
\def\E{{\mathbb E}}

\newcommand{\M}{{\mathcal M}}

\newcommand{\BB}{{\mathcal B}}

\newtheorem{thm}{Theorem}

\newtheorem{lemma}{Lemma}
\newtheorem{prop}{Proposition}
\newtheorem{definition}{Definition}

\begin{document}
\newcommand{\edt}[1]{{\vbox{ \hbox{#1} \vskip-0.3em \hrule}}}

\begin{center}
{\LARGE\bf Summary statistics for inhomogeneous marked point processes}\\[0.2in]

{\large O.\ Cronie and M.N.M.\ van Lieshout}\\[0.1in]

{\em CWI, 
P.O.\ Box 94079, 
NL-1090 GB Amsterdam, 
The Netherlands }
\end{center}

\noindent{\bf Abstract}: 
We propose new summary statistics for intensity-reweighted moment stationary 
marked point processes with particular emphasis on discrete marks. The new 
statistics are based on the $n$-point correlation functions and reduce to 
cross $J$- and $D$--functions when stationarity holds. We explore the
relationships between the various functions and discuss their explicit forms 
under specific model assumptions. We derive ratio-unbiased minus sampling 
estimators for our statistics and illustrate their use on a data set of 
wildfires.




\section{Introduction}
\label{S:intro}

The analysis of a marked point pattern typically begins with computing 
some summary statistics which may be used to find specific structures 
in the data and suggest suitable models
\cite{SKM,DVJ1,DVJ2,Handbook,Illian,MCbook}. The choice of summary
characteristic depends both on the pattern at hand and on the feature or 
hypothesis of interest. Indeed, under the working assumption of 
stationarity, for discrete marks, cross versions of the $K$- or
nearest neighbour distance distribution function may be appropriate 
\cite{Digglebook}; for real-valued marks, the mark correlation functions 
of \cite{Penttinen} are widely used. Various types of $J$-functions 
\cite{MCJfunMPP,MCBaddeley2} offer useful alternatives.

Often, however, the assumption of homogeneity cannot be justified. In the 
unmarked case, \cite{BaddeleyEtAl} proposed an inhomogeneous extension of 
the $K$-function for so-called second order intensity-reweighted stationary 
point processes. Their ideas were extended to spatio-temporal point processes
in \cite{GabrielDiggle,MollerGhorbani}, whereas \cite{CronLies13,MCJfun} 
extended the $J$-function under the somewhat stronger assumption 
of intensity-reweighted moment stationarity in space and time.

For non-stationary multivariate point processes, 
\cite{MollerBook} proposed an extension of the $K$-function under the 
assumption of second order intensity-reweighted stationarity.
As we will indicate in this paper, this structure may be extended to 
$K$-functions for general marked point processes.

Regarding $J$-functions, 
in \cite{MCJfun} the author noted that the ideas in that paper could be 
combined with those in \cite{MCJfunMPP} to define inhomogenous $J$-functions 
with respect to mark sets. In this paper we do so, and, as a by-product,
obtain a generalisation of the cross nearest-neighbour distance
distribution function.

The paper is structured as follows. 
In Section~\ref{Preliminaries}, we define marked point processes with 
locations in Euclidean spaces and give the necessary preliminaries. 
In Sections~\ref{SectionFGfun} and \ref{SectionJfun}, 
we define, respectively, cross $D$- and $J$-functions for 
inhomogeneous multivariate point processes and propose generalisations 
to point processes with real-valued marks. We show that $D$ and $J$ can
be expressed in terms of the generating functional and discuss the
relationships between these statistics and the cross $K$-function. 
In Section~\ref{SectionIndependence}, we investigate the form of our 
statistics under various independence and marking assumptions.
We derive minus sampling estimators in Section~\ref{SectionEstimation}, 
which are applied to a data set on wildfires in New Brunswick, Canada, 
in Section~\ref{SectionApplications}. We finish the paper with a summary.

\section{Definitions and notations}
\label{Preliminaries}

Throughout this paper, we consider marked point processes $Y$
\cite[Definition~6.4.1]{DVJ1} with points in $\R^d$ equipped with the 
Euclidean metric and Borel $\sigma$-algebra $\BB(\R^d)$. We write
$\ell$ for the Lebesgue measure on $\BB(\R^d)$. By definition, 
the ground process $Z$ obtained from $Y$ by ignoring the marks is a 
well-defined point process on $\R^d$ in its own right. We shall 
assume that $Z$ is simple, that is, almost surely does not contain
multiple points. 

We assume that the mark space $\M$ is Polish and equipped with a finite 
reference measure $\nu$ on the Borel $\sigma$-algebra $\BB(\M)$. We denote 
by $\BB(\R^{d}\times\M)$ the Borel $\sigma$-algebra on the product space 
$\R^{d}\times\M$. In the special case that $\M$ is finite, $Y$ can be seen
as a multivariate point process $(Y_1, \dots, Y_k)$ where $Y_i$ contains
the points marked $i \in \M = \{ 1, \dots, k \}$. 

\subsection{Product densities}
\label{S:product}

Recall that the intensity measure of a marked point process 
is defined on product sets $A = B\times C \in \BB(\R^{d}\times\M)$
by 
\[
\Lambda(A) = \E Y(A) = \E Y(B\times C),
\]
the expected number of points in $B$ with marks in $C$. If $\Lambda$ is 
locally finite as a set-function, it can be extended to a measure on 
$\BB(\R^{d}\times\M)$ (see e.g.\ \cite[Theorem A, p. 54]{Halmos}). In 
this paper, additionally, we assume that $\Lambda$ admits a density 
$\lambda$ with respect to $\ell \times \nu$, which is referred to as 
the intensity function. In particular, for a finite mark space, 
$\lambda(z,i) \nu(i) = \lambda_i(z)$ is the intensity function of $Y_i$.

Since for fixed $C$ the measure $\Lambda(\cdot \times C)$ is 
absolutely continuous with respect to the intensity measure $\Lambda_g$ 
of the ground process,
\begin{align}
\label{e:intensity}
\Lambda(B\times C) = \int_{B} M^{z}(C) \, \Lambda_g(dz).
\end{align}
Here $M^z(C)$ is the probability that the mark of a point at location
$z$ falls in $C$. The members of the family $\{ M^z: z \in \R^d \}$ of 
probability distributions on the Borel sets of $\M$ are called
{\em mark distributions}. 

If $Y$ is stationary, that is, if its distribution is invariant under 
translations of the locations, 
\(
\Lambda(B\times C) = \lambda \nu_M(C) \ell(B)
\)
for some probability distribution $\nu_M$ on $\M$, which is known as 
the {\em mark distribution}. In this case, we may take $\nu = \nu_M$ for the 
reference measure on $\M$ so that $\Lambda$ has constant intensity function 
$\lambda$ with respect to $\ell \times \nu_M$, and, moreover, 
$\lambda$ is the intensity of the ground process. 

Higher order `intensity functions' or \emph{product densities} can be 
defined as densities $\rho^{(n)}$ of the factorial moment measures 
provided these exist, in which case they satisfy the following $n$-th
order \emph{Campbell formula}. For any measurable function $f\geq 0$, 
the sum of $f$ over $n$-tuples of different points of $Y$ is a random 
variable with expectation
\begin{align} 
\label{CampbellMPP}
&\E\left[
\sum_{(z_1,m_1), \ldots, (z_n,m_n) \in Y}^{\neq}
f((z_1,m_1), \ldots, (z_n,m_n))
\right]
=
\\
&=
\int \cdots \int 
f((z_1,m_1), \ldots, (z_n,m_n)) \,
\rho^{(n)}((z_1,m_1), \ldots, (z_1,m_1)) \prod_{i=1}^n dz_i d\nu(m_i)
\nn
\end{align}
(with the left hand side being infinite if and only if the right hand side
is infinite). Note that $\rho^{(1)} = \lambda$, the intensity function. Also, 
$n$-point mark distributions $M^{z_1, \dots, z_n}(C_1 \times \cdots\times C_n)$
can be defined analogously to the case $n=1$. For further details, see 
for example the textbook \cite{SKM}. Note that, by the absolute continuity 
underlying the existence of $\rho^{(n)}$, there exist product densities 
$\rho_g^{(n)}(z_1, \ldots, z_n)$ for the ground process and densities 
$f_{z_1, \dots, z_n}$ of $M^{z_1, \dots, z_n}$ with respect to the $n$-fold product
of $\nu$ with itself such that
\[
M^{z_1, \dots, z_n}(C_1 \times \cdots \times C_n) =
\int_{C_1 \times \cdots \times C_n} f_{z_1, \dots, z_n}(m_1, \dots, m_n)
\prod_{i=1}^n d\nu(m_i).
\]
In particular, the intensity function of the ground process is given 
by $\lambda_g(z) = \rho_g^{(1)}(z)$ and $\lambda(z,m) = f_z(m) \lambda_g(z)$.

We will also need the related concept of {\em $n$-point correlation functions}
$\xi_n$, $n\geq1$, the intensity-reweighted densities of the factorial 
cumulant measures \cite[Section~9.5]{DVJ2}. These permutation invariant 
measurable functions are defined by the following recursive relation
(see e.g.\ \cite{MCJfunMPP,White}). Set $\xi_1 \equiv 1$ and, for $n\geq2$, 
\begin{align}
\label{CorrFuncMPP}
\sum_{k=1}^{n}
\sum_{E_1,\ldots,E_k}
\prod_{j=1}^{k}
\xi_{|E_j|}(\{(z_i,m_i):i\in E_j\})
&=
\frac{\rho^{(n)}((z_1,m_1),\ldots,(z_n,m_n))}{
\lambda(z_1,m_1) \cdots \lambda(z_n,m_n)},
\end{align}
where $\sum_{E_1,\ldots,E_k}$ is a sum over all possible $k$-sized partitions 
$\{E_1,\ldots,E_k\}$, $E_j\neq \emptyset$, of the set $\{1,\ldots,n\}$ and 
$|E_j|$ denotes the cardinality of $E_j$. 
Note that for a Poisson process, $\xi_n\equiv0$ for all $n\geq2$. 

\subsection{Palm measures and conditional intensities}
\label{S:semiPalm}

Let $Y$ be a simple marked point process whose intensity function exists. 
The summary statistics in this paper are defined in terms of {\em reduced Palm 
measures} satisfying the \emph{reduced Campbell-Mecke} formula which states 
that, for any measurable function $f \geq 0$,
\begin{align} 
\label{CMthmMPP}
&\E\left[ \sum_{(z,m) \in Y} f( (z,m), Y \setminus \{(z,m) \}) \right]
=
\int_{\R^d} \int_{\M}
\E^{!(z,m)} \left [f( (z,m), Y ) \right]
\lambda(z,m) \, dz d\nu(m)
\end{align}
(with the left hand side being infinite if and only if the right hand 
side is infinite). The probability measure $P^{!(z,m)}$ corresponding to 
$\E^{!(z,m)}$ can be interpreted as the conditional probability of 
$Y\setminus \{ (z,m) \}$ given that $Y(\{(z,m)\}) = 1$. 
For further details see \cite{DVJ2}.

A few remarks are in order. First, consider the special case that $Y$
is stationary and the reference measure on $\M$ is the mark distribution 
$\nu_M$. In this case, it is possible to define reduced Palm measures with 
respect to arbitrary mark sets. Specifically, for $C\in\BB(\M)$ such that 
$\nu(C) = \nu_M(C) > 0$, set
\begin{equation}
\label{e:Palm}
P_{C}^{!z}(R)
=
\frac{1}{\nu(C)}
\int_{C} P^{!(z,m)}(R) \, d\nu(m).
\end{equation}
Then, $P_C^{!z}$ does not depend on the choice of $z\in\R^d$ and is a 
probability measure \cite[Section~4.4.8]{SKM}. It can be interpreted as 
the conditional distribution of $Y$ on the complement of $ \{ z \} \times\M$, 
given that $Y$ places a point at $z$ with mark in $C$.

As a second example, consider multivariate point processes $(Y_1, \dots,
Y_k)$ and let $\nu$ be any finite measure on $\M = \{ 1, \dots, k \}$.
Now, we have a family of reduced Palm measures $P^{!(z,i)}$ for $i=1, \dots, k$
and we will restrict ourselves to sets of the form $C = \{ i \}$. Then
(\ref{e:Palm}) reads 
\[
P_C^{!z}(R) = \frac{1}{\nu(i)} \nu(i) P^{!(z,i)}(R) = P^{!(z,i)}(R) 
\]
and does not depend on the specific choice of $\nu$.

For non-finite mark spaces, the reference measure $\nu$ on $\M$ may not 
correspond to a well-defined mark distribution. One pragmatic approach is to 
take a finite partition of the mark space, $\M = \cup_{i=1}^k \M_i$, and 
proceed as in the multivariate case. 
An alternative is to use (\ref{e:Palm}) as definition for a $\nu$-averaged 
reduced Palm distribution with respect to $C$, bearing in mind that the 
definition does depend on the choice of $\nu$.

\subsection{Generating functionals}
\label{S:genFunc}

When product densities of all orders exist, the {\em generating functional}
$G(\cdot)$, which unique\-ly determines the distribution of $Y$ (see e.g.\ 
\cite[Thm 9.4.V.]{DVJ2}), is defined as follows.
For all mappings $v = 1 - u$ such that $u: \R^{d} \times \M \rightarrow
[0,1]$ is measurable with bounded support, set
\begin{align}
&\label{GFMPP}
G(v) 
= G(1-u)
= \E\left[ \prod_{(z,m) \in Y} v(z,m) \right]
\\
&=
1 + \sum_{n=1}^{\infty} \frac{(-1)^n}{n!} 
\int_{\R^d \times \M} \cdots \int_{\R^d \times \M} 
\rho^{(n)}((z_1,m_1), \ldots, (z_n,m_n)) 
\prod_{i=1}^{n} u(z_i, m_i) \, dz_i d\nu(m_i)
\nn
\\
&=
\exp\left[
\sum_{n=1}^{\infty} \frac{(-1)^n}{n!} 
\int_{\R^d \times \M} \cdots \int_{\R^d \times \M} 
\xi_n( (z_1,m_1), \ldots, (z_n,m_n) ) \right. \times
\nn
\\
&
\quad \quad \quad \quad \quad \quad \quad \quad \quad
\times \left. \prod_{i=1}^{n} u(z_i, m_i) \lambda(z_i, m_i) \, dz_i d\nu(m_i)
\right].
\nn
\end{align}
By convention, $\log 0 = -\infty$ and an empty product equals $1$. 
The last equalities holds provided that the right hand sides converge 
(see e.g.\ \cite[p.~126]{SKM}). 
Similarly, for $a\in\R^d$ and $C\in\BB(\M)$, we may define the 
generating functional $G^{!a}_{C}$ with respect to $P^{!a}_{C}$ 
(see the discussion around (\ref{e:Palm}))
by 
\begin{equation}
\label{PalmGFMPP}
G^{!a}_{C}(v) 
=
\frac{1}{\nu(C)} \int_{C} 
\E^{!(a,b)} \left[ \prod_{(z,m) \in Y} v(z,m) \right]
d\nu(b).
\end{equation}

\section{Definition of summary statistics}
\subsection{Inhomogeneous cross $D$-function}
\label{SectionFGfun} 

In this section, we define cross $D$-functions for marked point
processes in analogy with the inhomogeneous nearest neighbour
distance distribution function of \cite{MCJfun}. Write 
\[
\bar \lambda_D = \inf_{z\in\R^d, m\in D} \lambda( z, m ).
\]
Throughout we assume that $Y$ is a simple 
marked point process whose product densities of all orders 
exist and for which the $\xi_n$, $n\geq 2$, are translation 
invariant in the sense
that 
\[
\xi_n( (z_1 + a, m_1), \dots,  (z_n + a, m_n) ) = 
\xi_n( (z_1, m_1), \dots, (z_n, m_n))
\]
for all $a\in \R^d$ and $\ell\otimes\nu$-almost all $(z_i, m_i) \in
\R^d \times \M$. 
If, moreover, $\bar \lambda = \bar \lambda_{\M} > 0$, 
then $Y$ is said to be {\em intensity-reweighted moment stationary}
(IRMS).

\begin{definition} \label{DefinitionDfunction}
Let $Y$ be IRMS and let $C$ and $D$ be Borel sets in $\M$ with $\nu(C)$ and 
$\nu(D)$ strictly positive. Write $B(a,r)$ for the closed ball centred
at $a$ with radius $r$. Set
\beann
u_r^{a,D}(z,m) 
= \frac{\bar\lambda_D \1\{(z,m)\in B(a,r) \times D\}}{\lambda(z,m)},
\quad
a\in\R^d,
\quad
D\in\BB(\M),
\eeann 
and define, for $r\geq 0$, the \emph{inhomogeneous cross 
nearest neighbour distance distribution function} by
\bea
\label{Ginhom}
D_{\rm inhom}^{CD}(r) &=& 1 - G^{!0}_{C}(1-u_r^{0,D}) 
\\
&=& 1 - 
\frac{1}{\nu(C)} \int_{C} 
\E^{!(0,b)} \left[
\prod_{(z,m) \in Y}
\left( 1 - \frac{\bar\lambda_D \1\{ (z,m) \in B(0,r) \times D \}
}{\lambda(z,m)} \right)
\right]
d\nu(b).
\nn
\eea
\end{definition}

We shall show in Theorem~\ref{ThmMarkedJfun} below that the specific 
choice $a=0$ in (\ref{Ginhom}) is merely a matter of convenience.
Moreover, $\bar \lambda_D$ may be replaced by smaller strictly positive
scalars.

When $Y$ is stationary and $\nu = \nu_M$, the mark distribution, 
\begin{eqnarray*}
D^{CD}_{\rm{inhom}}(r) & = & 1 - 
\frac{1}{\nu_M(C)} \int_{C} \E^{!(0,b)} \left[ \prod_{(z,m) \in Y}
\1\{ (z,m) \notin B(0,r) \times D\}\right] d\nu_M(b) \\
& = & P^{!0}_C ( Y \cap B(0,r) \times D \neq \emptyset ),
\end{eqnarray*}
so that (\ref{Ginhom}) reduces to the $C$-to-$D$ nearest neighbour
distance distribution for marked point processes \cite{MCJfunMPP}.

\subsubsection{Multivariate point process}

Consider a multivariate point process $Y = (Y_1, \dots, Y_k)$ that is
intensity-reweighted moment stationary. Let $C = \{ i \}$ and
$D = \{ j \}$ for $i\neq j \in \{1, \dots, k\}$. Write $\bar \lambda_j
= \inf_{z\in\R^d} \lambda_j(z)$ and note that $\bar \lambda_D / \lambda(z,j)$
is equal to $\bar \lambda_j / \lambda_j(z)$. Therefore (\ref{Ginhom})
reduces to
\bea
\label{GinhomMulti}
D^{ij}_{\rm{inhom}}(r) = 1 - \E^{!(0,i)} \left[ \prod_{z \in Y_j} 
\left( 1 - \frac{\bar \lambda_j}{\lambda_j(z)} 
1\{ z \in B(0,r) \} \right) \right]  
\eea
which under the further assumption that  $Y$ is stationary is equal to
\[
 P^{!(0,i)}( Y_j \cap B(0,r) \neq \emptyset ),
\]
the classical cross nearest neighbour distance distribution function,
see e.g. \cite[Chapter~21]{Handbook}. If $Y$ is a Poisson process,
\[
D^{ij}_{\rm{inhom}}(r)  = 1 - \exp\left[ - \bar\lambda_j \ell(B(0,r)) \right].
\]
Smaller values of $D^{ij}_{\rm inhom}(r)$ suggest there are fewer points of 
type $j$ in the $r$-neigh\-bour\-hood, that is, inhibition; larger values indicate
that points of type $j$ are attracted by those of type $i$ at range $r$.
In the case $i=j$, we obtain the inhomogeneous $D$-function of $Y_i$. 

With $C = \{ i \}$ for some $i \in \{ 1, \dots, k \}$ and $D = \M =
\{ 1, \dots, k \}$, (\ref{Ginhom}) is equal to
\begin{equation}
\label{GinhomAny}
D^{i\bullet}_{\rm{inhom}}(r) = 1 -
\E^{!(0,i)}\left[
\prod_{(z,m)\in Y}
\left(1-\frac{\bar\lambda \1\{ z \in B(0,r) \}}{\lambda(z,m)}\right)
\right]
\end{equation}
for $r\geq 0$. Note that the function $u^{0,\M}_r$ may depend on 
$\nu$ through $\lambda(z,m)$. If we give equal weight to each member of 
$\M$, however, 
$\bar \lambda / \lambda(z,m) = \tilde \lambda / \lambda_m(z)$ is 
uniquely defined in terms of the intensity functions of the components
of $Y$ and the minimal marginal intensity $\tilde \lambda =
\inf \{ \lambda_i(z): z \in \R^d, i \in \{ 1, \dots, k \} \}$. 
If $Y$ is stationary, $D^{i\bullet}$ is the classic $i$-to-any nearest 
neighbour distance distribution.

\subsection{Inhomogeneous cross $J$-functions}\label{SectionJfun}

In this section, we define cross $J$-functions for marked point
processes in analogy with the inhomogeneous $J$-function of 
\cite{MCJfun}. Throughout we assume that $Y$ is a simple 
intensity-reweighted moment stationary point process.

\begin{definition}\label{JfunMPPdef}
Let $Y$ be IRMS and let $C$ and $D$ be Borel sets in $\M$ with $\nu(C)$ 
and $\nu(D)$ strictly positive. For $r\geq 0$ and $n \geq 1$, set
\beann
J_n^{CD}(r) 
&=&
\int_{C}
\int_{B(0,r) \times D} \cdots \int_{B(0,r) \times D} 
\xi_{n+1}((a,b), (z_1 + a, m_1), \ldots, (z_n + a, m_n)) 
\\
&&\times d\nu(b) \prod_{i=1}^n dz_i d\nu(m_i)
\eeann
and define the {\em inhomogeneous cross $J$-function} by
\begin{align}
\label{JfunMPP}
J_{\rm inhom}^{CD}(r) = 
\frac{1}{\nu(C)} \left( \nu(C) + 
\sum_{n=1}^{\infty} \frac{(-\bar \lambda_D)^n}{n!} J_n^{CD}(r) \right)
\end{align}
for all ranges $r \geq 0$ for which the series is absolutely convergent. 
\end{definition}

Note that there is an implicit dependence on $a \in \R^d$ in $J_n^{CD}(r)$ 
and consequently in $J_{\rm inhom}^{CD}(r)$. However, the IRMS assumption 
implies that all $J_n^{CD}(r)$ (and therefore $J_{\rm inhom}^{CD}(r)$) are 
$\ell$-almost everywhere constant. Furthermore, 
Cauchy's root test implies that whenever
\(
\limsup_{n\rightarrow\infty}
\left(
\frac{\bar\lambda_D^n}{n!} |J_n^{CD}(r)| \right)^{1/n} < 1,
\)
(\ref{JfunMPP}) is absolutely convergent.

When $Y$ is stationary and $\nu=\nu_M$, the mark distribution,
(\ref{JfunMPP}) reduces to the cross inhomogeneous $J$-function
for marked point processes introduced in \cite{MCJfunMPP} since in that
case $\bar\lambda_D = \bar \lambda_\M $ regardless of the choice of $D$. 
Finally, note that for a Poisson process, $\xi_n \equiv 0$ for 
$n\geq 2$, so $J_{CD}(r) \equiv 1$.
In general, the inhomogeneous $J$-function is not commutative with respect 
to the mark sets $C$ and $D$, $C\neq D$.

Looking closer at Definition \ref{JfunMPPdef}, we see that there is some 
resemblance between $J_{\rm inhom}^{CD}(r)$ and the cross inhomogeneous 
$K$-function defined in \cite[Def.~4.8]{MollerBook}. Indeed, 
truncation of the series in (\ref{JfunMPP}) at $n=1$ gives
\begin{eqnarray*}
J_{\rm inhom}^{CD}(r) -1
& \approx  &
- \frac{\bar\lambda_D}{\nu(C)} \int_C \int_{B(0,r)\times D} 
\left[
g( (0,b), (z,m) ) - 1
\right] d\nu(b) dz d\nu(m) \\
& = &
- \bar \lambda_D \nu(D) \left(
      K_{\rm inhom}^{CD}(r) - \ell(B[0,r])
\right),
\end{eqnarray*}
where 
\begin{equation}
\label{Kinhom}
K_{\rm inhom}^{CD}(r) = \frac{1}{\nu(C) \nu(D)} \int_C \int_{B(0,r)\times D}
g((0, m_1), (z, m_2) ) d\nu(m_1) dz d\nu(m_2)
\end{equation}
is the generalisation of the cross inhomogeneous $K$-function to our
set-up. 
Note that the inhomogeneous $K$-function defined by (\ref{Kinhom}) requires 
translation invariance of the two-point correlation function only, in which 
case $Y$ is said to be second-order intensity reweighted stationary (SOIRS). 
Heuristically, $J_{CD}(r) < 1$ suggests that points with marks in $D$ tend 
to cluster around points with marks in $C$ at range $r\geq 0$; $J_{CD}(r)>1$ 
indicates that points with marks in $D$ avoid those with marks in $C$ at 
range $r\geq 0$. This interpretation is confirmed by 
Theorem~\ref{ThmMarkedJfun} below.

Definition~\ref{JfunMPPdef} is hard to work with. A more natural
representation can be given in terms of the generating functional.
In order to do so, define the \emph{inhomogeneous empty space function} 
of $Y_D$, the marked point process $Y$ restricted to $\R^d \times D$, by
\bea
\label{Finhom}
1 - F_{\rm inhom}^{D}(r) = G(1-u_r^{a,D})
=
\E\left[
\prod_{(z,m) \in Y}
\left(1-\frac{\bar\lambda_D \1\{ (z,m) \in B(a,r) \times D \}}{
\lambda(z,m)} \right)
\right]
\eea
under the convention that empty products equal 1 and with $u_r^{a,D}$ as
in Definition~\ref{DefinitionDfunction}. As for $D_{\rm inhom}^{CD}$,
the definition does not depend on the choice of origin $a=0$ and 
$\bar \lambda_D$ may be replaced by smaller strictly positive scalars.
At this point, it is important to stress that for $D = \M$, 
$F_{\rm inhom}^D$ is not necessarily equal to $F_{\rm inhom}$, the 
empty space function of the ground process $Z$, since $u_r^{0, \M}$
depends on the marks both through the intensity function $\lambda$ and
the bound $\bar \lambda_\M$. 

\begin{thm}\label{ThmMarkedJfun}
Let $Y$ be as in Definition~\ref{JfunMPPdef}. Then, as a function 
of $a\in\R^d$, each $J_n^{CD}(r)$ is $\ell$-almost everywhere constant.
Moreover, if 
\[
\limsup_{n\rightarrow\infty} \left( 
\frac{(-\bar\lambda_D)^n}{n!} 
\int_{B(0,r) \times D} \cdots \int_{B(0,r) \times D}
\frac
{\rho^{(n)}((z_1, m_1), \ldots, (z_n, m_n))}
{\lambda(z_1,m_1) \cdots \lambda(z_n, m_n)}
\prod_{i=1}^{n} dz_i d\nu(m_i)
\right)^{1/n} 
\]
is strictly less than $1$, then, for almost all $a\in\R^d$, 
the $C$-to-$D$ inhomogeneous $J$-function of Definition~\ref{JfunMPPdef} 
satisfies
\beann
J_{\rm inhom}^{CD}(r) = 
\frac{1-D_{\rm inhom}^{CD}(r)}{1-F_{\rm inhom}^{D}(r)}
\eeann
for all $r\geq0$ for which $F_{\rm inhom}^{D}(r)\neq1$.
\end{thm}

The proof is technical and relegated to Appendix~A. 

\subsubsection{Multivariate point process}

Consider a multivariate point process $Y = (Y_1,  \dots, Y_k)$ that is 
intensity-reweighted moment stationary. By a suitable choice of mark set $D$,  
we obtain different types of inhomogeneous $J$-functions. 

First, take $C = \{ i \}$ and $D = \{ j \}$ for $i\neq j \in \{1, \dots, k\}$. 
Then, writing $F^j_{\rm inhom}$ for the inhomogeneous empty space function of 
$Y_j$ and recalling (\ref{GinhomMulti}), the statistic (\ref{JfunMPP}) is 
equal to 
\bea
\label{JfunMulti}
J^{ij}_{{\rm inhom}}(r) = \frac{ 1 - D_{\rm{inhom}}^{ij}(r) 
}{ 1 - F^j_{\rm{inhom}}(r)}
\eea
and compares the distribution of intensity-reweighted distances from a 
point of type $i$ to the nearest one of type $j$ to those from an arbitrary 
point to $Y_j$. Therefore, it generalises the $i$-to-$j$ cross $J$-function of
\cite{MCBaddeley2} for stationary multivariate point processes.

Set $C = \{ i \}$ for some $i \in \{ 1, \dots, k \}$ and $D = \M =
\{ 1, \dots, k \}$. Then, recalling (\ref{GinhomAny}), the statistic 
(\ref{JfunMPP}) can be written as 
\begin{equation}
\label{JfunAny}
J^{i \bullet}_{\rm{inhom}}(r)  =
\frac{1 - D^{i\bullet}_{\rm{inhom}}(r)}{ 1 - F^\M_{\rm{inhom}}(r) }
\end{equation}
and compares tails of the $i$-to-any nearest neighbour distance distribution
and the empty space function of $Y$. Note that if $\nu$ is proportional
to the counting measure, $F^\M_{\rm{inhom}}(r)$ can be expressed in terms of
the intensity functions of the components and the minimal marginal 
intensity (see the discussion following formula (\ref{GinhomAny})).  Hence,
$J^{i \bullet}_{\rm inhom}$ generalises the $i$-to-any $J$-function for 
stationary multivariate point processes \cite{MCBaddeley2}.

\section{Independence and random labelling}
\label{SectionIndependence}

In this section, we investigate the effect of various independence 
assumptions and marking schemes on our summary statistics. 

\subsection{Independent marking mechanisms}

Specific forms of marking are summarised in Definition~\ref{D:marking} below
\cite[Definition~6.4III]{DVJ1}.

\begin{definition} 
\label{D:marking} 
A marked point process $Y$
is called {\em independently marked} if, given the ground process
$Z$, the marks are independent random variables with a distribution 
that depends only on the corresponding location.
If, additionally, $M^z$ does not depend on the location, 
we say that $Y$ has the {\em random labelling property}. 
\end{definition}

\begin{prop}
\label{LemmaIndependentlyMarked}
Let $C$ and $D$ be Borel sets in $\M$ with $\nu(C),\nu(D)>0$ and
assume that $Y$ is independently marked. 
\begin{itemize}
\item[a)] 
If $Y$ is SOIRS, the ground process $Z$ is also SOIRS and
$K_{\rm inhom}^{CD}(r) = K^Z_{\rm inhom}(r)$, the inhomogeneous
$K$-function of $Z$.
\end{itemize}
Let $\E_Z^{!0}$ denote the expectation under the Palm distribution of the 
ground process $Z$ and write $u_r^0(z) = \bar \lambda_g \1\{ z \in B(0,r) \} /
\lambda_g(z)$, $z\in\R^d$, $c_D = \bar \lambda_D \nu(D) / \bar \lambda_g$.
Under the assumptions of Theorem~\ref{ThmMarkedJfun}, when $Y$ is 
independently marked, $Z$ is IRMS and
\begin{itemize}
\item[b)] 
$F_{\rm inhom}^{D}(r) = 1 - G_Z(1- c_D u_r^0 )$,
\item[c)] 
$D_{\rm inhom}^{CD}(r) = 1 - G_Z^{!0}(1- c_D u_r^0 )$,
\item[d)] 
$J_{\rm inhom}^{CD}(r) = G_Z^{!0}(1- c_D u_r^0 ) / G_Z(1- c_D u_r^0)$ 
for all $r\geq 0$ for which the denominator is non-zero.
\end{itemize}
If $Y$ is randomly labelled with $\nu = \nu_M$,  
then $c_D = \nu_M(D)$ and
\[
F_{\rm inhom}^{\M}(r) = F_{\rm inhom}^Z(r); \quad
D_{\rm inhom}^{C\M}(r) = D_{\rm inhom}^Z(r); \quad
J_{\rm inhom}^{C\M}(r) = J_{\rm inhom}^Z(r). 
\]
\end{prop}

\

\begin{proof}
Recall that 
\[
\rho^{(n)}( (z_1, m_1), \dots, (z_n, m_n) ) = 
f_{z_1, \dots, z_n}(m_1, \dots, m_n) \rho_g^{(n)}( z_1, \dots, z_n ).
\]
Under the independent marking assumption, 
\begin{equation}
\label{e:prodmark}
f_{z_1, \dots, z_n}(m_1, \dots, m_n) = \prod_{i=1}^n f_{z_i}(m_i).
\end{equation}
Therefore,
\[
\xi_n( (z_1, m_1), \dots, (z_n, m_n) ) = \xi^g_n( z_1, \dots, z_n ),
\]
the $n$-point correlation function of the ground process, so that $Z$ is 
(second order) intensity-reweighted moment stationary whenever $Y$ is.
Plugging (\ref{e:prodmark}) into (\ref{Kinhom}) yields
$K_{\rm inhom}^{CD}(r) = \int_{B(0,r)} g_g(0,z) \, dz $, the inhomogeneous
$K$-function of $Z$. 
Furthermore, under the assumption that the series expansion is
absolutely convergent, by (\ref{GFMPP}), (\ref{Finhom}) reduces to 
\begin{eqnarray*}
1 - F_{\rm inhom}^{D}(r) & = & 1 + \sum_{n=1}^\infty \frac{(-\bar \lambda_D)^n}{n!}
\int_{( B(0,r)\times D )^n}
\frac{\rho_g^{(n)}( z_1, \dots z_n)}{\lambda_g(z_1) \cdots \lambda_g(z_n)}
\prod_{i=1}^n dz_i d\nu(m_i) \\
& = & 1 + \sum_{n=1}^\infty \frac{(-\bar \lambda_D \nu(D))^n}{n!}
\int_{ B(0,r)^n}
\frac{\rho_g^{(n)}( z_1, \dots z_n)}{\lambda_g(z_1) \cdots \lambda_g(z_n)}
\, \prod dz_i.
\end{eqnarray*}
Similarly,
\[
1 - D_{\rm inhom}^{CD}(r) =  
1 + \sum_{n=1}^\infty \frac{(-\bar \lambda_D \nu(D))^n}{n!}
\int_{ B(0,r)^n}
\frac{\rho_g^{(n+1)}(0, z_1, \dots z_n)}{
\lambda_g(0) \lambda_g(z_1) \cdots \lambda_g(z_n)}
\, \prod dz_i.
\]
We conclude that $1 - D_{\rm inhom}^{CD}(r) = G_Z^{!0}(1 - \bar\lambda_D \nu(D) 
\1_{B(0,r)}(\cdot) / \lambda_g(\cdot))$ and $1 - F_{\rm inhom}^D(r) = 
G_Z(1 - \bar\lambda_D \nu(D) \1_{B(0,r)}(\cdot) / \lambda_g(\cdot))$.

Under random labelling, the right hand side of (\ref{e:prodmark}) is 
further simplified to $\prod_{i=1}^n f(m_i)$ for some 
probability density $f$ that does not depend on location. 
 If furthermore $\nu = M^z = \nu_M$, 
the mark distribution, the density is one, i.e.\ $f(m_i) \equiv 1$. Hence 
$\bar \lambda_D = \bar \lambda_g$ and $c_D = \nu_M(D)$. In particular $c_\M = 1$.
\end{proof}

Note that the summary statistics do not depend on the choice of $C$, 
but may depend on $D$ through $c_D$. 

\subsection{Independence}

Recall that we use the notation $Y_C$, $C\in\BB(\M)$, for the restriction of
$Y$ to $\R^d \times C$. If $Y_C$ and $Y_D$ are independent, then the $C$-to-$D$
cross $J$-function is identically $1$. More precisely, the following result
holds.

\begin{prop}
\label{LemmaIndependence1}
Consider two disjoint Borel sets $C, D \subseteq \M$ with $\nu(C)$ and
$\nu(D)$ strictly positive and assume that $Y_C$ and $Y_D$ are independent. 
Under the assumptions of Theorem \ref{ThmMarkedJfun}, 
\[ 
D_{\rm inhom}^{CD}(r)  = F_{\rm inhom}^D(r)
\]
so that $J_{\rm inhom}^{CD}(r)\equiv 1$ whenever well defined. 
\end{prop}

If $Y$ is SOIRS,  $K_{\rm inhom}^{CD}(r) = \omega_d r^d = \ell(B(0,r))$ 
whenever $Y_C$ and $Y_D$ are independent by \cite[Proposition~4.4]{MollerBook}.

\begin{proof}

By the Campbell formula (\ref{CampbellMPP}), 
if $Y_C$ and $Y_D$ are independent, the product densities factorise with
respect to $C$ and $D$, i.e.
\begin{align*}
& \rho^{(n_C + n_D)}( (z_1, m_1), \dots, (z_{n_C}, m_{n_C}),
(\tilde z_1, \tilde m_1), \dots, (\tilde z_{n_D}, \tilde m_{n_D})) 
\\
= &\rho^{(n_C)}((z_1, m_1), \dots, (z_{n_C}, m_{n_C}) ) \,
\rho^{(n_D)}( (\tilde z_1, \tilde m_1), \dots, (\tilde z_{n_D}, \tilde m_{n_D})) 
\end{align*}
for almost all $(z_i, m_i) \in \R^d \times C$ and $(\tilde z_i, \tilde m_i)
\in \R^d \times D$. Then, by the proof of Theorem~\ref{ThmMarkedJfun}, 
\begin{align*}
&G^{!0}_{C}(1-u_r^{0,D})
=
1+
\frac{1}{\nu(C)} 
\times
\\
&
\sum_{n=1}^{\infty}
\frac{(-\bar\lambda_D)^n}{n!}
\int_{C} 
\bigg(
\int_{(B(0,r)\times D)^n}
\frac{
\rho^{(n+1)}((0,b),(z_1,m_1),\ldots,(z_n,m_n))
}{\lambda(0,b) \lambda(z_1,m_1) \cdots \lambda(z_n,m_n)}
\prod_{i=1}^{n} dz_i d\nu(m_i)
\bigg)
d\nu(b).
\end{align*}
The integrand factorises as 
\[
\frac{ \lambda(0,b) \rho^{(n)}((z_1,m_1),\ldots,(z_n,m_n))
}{\lambda(0,b) \lambda(z_1,m_1) \cdots \lambda(z_n,m_n)}
\]
so that
\(
G^{!0}_{C}(1-u_r^{0,D})
=
G(1-u_r^{0,D}).
\)
We conclude that 
$D_{\rm inhom}^{CD}(r) = F_{\rm inhom}^D(r)$ and $J_{\rm inhom}^{CD}(r)\equiv 1$. 
\end{proof}

Proposition~\ref{LemmaIndependence1} generalises well-known results 
for stationary multivariate point processes \cite{Digglebook,MCBaddeley2}.
The next result collects mixture formulae.

\begin{prop}
\label{LemmaIndependence2}
Let $C \subseteq \M$ be a Borel set with $0 < \nu(C) < \nu(\M)$.
Set $D = \M \setminus C$ and assume that $Y_C$ and $Y_D$ are independent. 
\begin{itemize}
\item[a)] If $Y$ is SOIRS,
 $K_{\rm inhom}^{C\M}(r) = \frac{\nu(D)}{\nu(\M)} \omega_d r^d +
\frac{\nu(C)}{\nu(\M)} K_{\rm inhom}^{Y_C}(r)$, where $K_{\rm inhom}^{Y_C}$ is 
the inhomogeneous $K$-function of $Y_C$. 
\end{itemize}
Write $c_A = \bar \lambda / \bar \lambda_A$ for $A\in\BB(\M)$.
Under the assumptions of Theorem \ref{ThmMarkedJfun}, 
\begin{itemize}
\item[b)]
$1 - F_{\rm inhom}^{\M}(r) 
=
G\left( 1 - c_C u_r^{0,C} \right)  
G\left( 1 - c_D u_r^{0,D} \right)
$, 
\item[c)]
$1 - D_{\rm inhom}^{C\M}(r) 
= 
G^{!0}_{C}\left( 1 - c_c u_r^{0,C} \right) 
G\left( 1 - c_D u_r^{0,D} \right) 
$,
\item[d)] 
$J_{\rm inhom}^{C\M}(r)
=
G^{!0}_{C}\left( 1 - c_C u_r^{0,C} \right)
\left/
G\left( 1 - c_C u_r^{0,C} \right)
\right.
$
for all $r\geq 0$ for which the denominator is non-zero. 
\end{itemize}
\end{prop}

Note here that if we would have used the global infimum $\bar \lambda$
in (\ref{Ginhom}), (\ref{Finhom}) and Definition \ref{JfunMPPdef}, 
the constants $c_C$ and $c_D$ would vanish and, e.g.,
$J_{\rm inhom}^{C \M}(r) \equiv J_{\rm inhom}^{Y_C}(r)$ whenever defined.

\begin{proof}
As in the proof of Proposition~\ref{LemmaIndependence1},
$g((0,m_1), (z, m_2)) = 1$ if $m_1 \in C$ and $m_2 \in \M \setminus C$, so that
\beann
K_{\rm inhom}^{C\M}(r) 
&=&
\frac{1}{\nu(C) \nu(\M)} \int_C \int_{B(0,r)\times \M}
g((0,m_1), (z, m_2)) \, d\nu(m_1) dz d\nu(m_2)
\\
&=&
\frac{\nu(\M\setminus C)}{\nu(\M)} \omega_d r^d 
+ 
\frac{\nu(C)}{\nu(\M)} K_{\rm inhom}^{CC}(r).
\eeann

Since $Y=Y_C\cup Y_{\M\setminus C}$ is the superposition of independent point
processes,
\beann
1 - F_{\rm inhom}^{\M}(r) 
&=& 
\E\left[
\prod_{(z,m)\in Y}
\left(1 - 
\frac{\bar{\lambda} \1\{(z,m)\in B(0,r)\times \M \}}{\lambda(z,m)}
\right)
\right]
\\
&=&
G\left(1 - \frac{\bar\lambda}{\bar\lambda_{\M\setminus C}}
u_r^{0,\M\setminus C}\right) 
G\left(1 - \frac{\bar\lambda}{\bar\lambda_{C}}
u_r^{0,C}\right).
\eeann
Furthermore, under each (conditional) distribution $P^{!(0,b)}$, $b\in C$, the 
points of $Y_{\M\setminus C}$ follow the distribution $P$, hence
\beann
D_{\rm inhom}^{C\M}(r) &=& 
1 - 
\frac{1}{\nu(C)} \int_{C} 
\E^{!(0,b)}\left[
\prod_{(z,m)\in Y}
\left(1-\frac{\bar\lambda \1\{(z,m)\in B(0,r) \times \M\}}{\lambda(z,m)}\right)
\right]
d\nu(b)
\\
&=& 
1 - 
\E\left[
\prod_{(z,m)\in Y \cap (B(0,r) \times \M \setminus C)}
\left(1 - \frac{\bar\lambda}{\lambda(z,m)} \right)
\right]
\times
\\
&&\times
\frac{1}{\nu(C)} 
\int_{C} 
\E^{!(0,b)}
\left[
\prod_{(z,m)\in Y \cap (B(0,r)\times C)}
\left(1 - \frac{\bar\lambda}{\lambda(z,m)} \right)
\right]
d\nu(b)
\\
&=& 
1 - 
G\left(1 - \frac{\bar\lambda}{\bar\lambda_{\M\setminus C}} 
u_r^{0,\M\setminus C}\right) 
G^{!0}_{C}\left(1 - \frac{\bar\lambda}{\bar\lambda_{C}} 
u_r^{0,C}\right).
\eeann
\end{proof}

\section{Statistical inference}
\label{SectionEstimation}

The goal of this section is to construct non-parametric estimators and 
tests.

Although defined on all of $\R^d$, in practice, 
the ground process $Z$ is observed only in some compact spatial
region $W \subseteq \R^d$ with boundary $\partial W$. In order to 
deal with edge effects, we apply a minus sampling scheme (see e.g.\ 
\cite{SKM} for further details). The underlying idea is that when one
is interested in the interactions up to range $r$, only observations 
in the eroded set
\[
W_{\ominus r} = \{ z \in W: d_{\R^d}(z, \partial W) \geq r \}
\]
are taken into account. For clarity of exposition, we assume that 
the intensity function $\lambda(z,m)$ is known. If this is not the 
case, an estimator $\hat{\lambda}$ may be plugged in. 

In Section~\ref{EstimateD}, we derive estimators for our inhomogeneous
$D$- and $J$-functions; estimation of the cross $K$-function is discussed in 
\cite{MollerBook}. Finally, Section~\ref{SectionTesting} is devoted 
to non-parametric tests for the independence and random labelling 
assumptions (cf.\ Section~\ref{SectionIndependence}).

\subsection{Estimation}
\label{EstimateD}

Let $Y$ be an intensity-reweighted moment stationary marked point process
and consider the estimators
\begin{equation}
\label{EstD}
\sum_{(a,b)\in Y\cap (W_{\ominus r} \times C)} 
\frac{1}{\lambda(a,b)} \left[
\prod_{(z,m)\in (Y\setminus \{(a,b)\})\cap(B(a,r)\times D)}
\left(1-\frac{\bar\lambda_D}{\lambda(z,m)}\right)
\right]
\end{equation}
and
\begin{equation}
\label{EstF}
\frac
{1}{|L\cap W_{\ominus r}|}
\sum_{l\in L\cap W_{\ominus r}}
\left[
\prod_{(z,m)\in Y \cap (B(l,r) \times D)}
\left(1-\frac{\bar\lambda_D}{\lambda(z,m)}\right)
\right],
\end{equation}
where $L\subseteq W$ is some finite point grid.
The following unbiasedness result holds.

\begin{lemma}
\label{L:EstD}
Suppose $C, D\in\BB(\M)$ have strictly positive $\nu$-content.
Then, under the assumptions of Theorem~\ref{ThmMarkedJfun},
provided $\ell(W_{\ominus r}) > 0$, (\ref{EstD}) and (\ref{EstF})
are unbiased estimators of, respectively,
\(
\ell(W_{\ominus r})\nu(C) G^{!0}_{C}(1-u_r^{0,D})
\)
and $G(1-u_r^{0,D})$.
\end{lemma}

\begin{proof}
By the Campbell--Mecke formula and Fubini's theorem,
the expectation of (\ref{EstD}) is equal to
\begin{align*}
&
\int_{W_{\ominus r}}\int_{C}
\E^{!(a,b)}
\left[
\frac{1}{\lambda(a,b)}
\prod_{(z,m)\in Y}
\left(1-\frac{\bar\lambda_D 
\1\{ (z,m) \in B(a,r) \times D \}}{\lambda(z,m)}\right)
\right]
\lambda(a,b) \, da d\nu(b)
\\
&=
\ell(W_{\ominus r})
\int_{C}
\E^{!(0,b)}
\left[
\prod_{(z,m)\in Y}
\left(1-\frac{\bar\lambda_D 
\1\{ (z,m) \in B(0,r) \times D \}}{\lambda(z,m)}\right)
\right]
d\nu(b),
\end{align*}
which can be written in generating functional terms as
 $\ell(W_{\ominus r})\nu(C) G^{!0}_{C}(1-u_r^{0,D})$.

The expectation of (\ref{EstF}) is 
\[
\frac
{1}{|L\cap W_{\ominus r}|}
\sum_{l\in L\cap W_{\ominus r}} 
G( 1 - u_r^{l,D}).
\]
Since the summands do not depend on $l$, the required unbiasedness follows.
\end{proof}

Lemma~\ref{L:EstD} implies that an estimator of $G^{!0}_{C}(1-u_r^{0,D})$ 
can be obtained from (\ref{EstD}) upon division by $\ell(W_{\ominus r})\nu(C)$. 
For irregular windows, however, the volume of $W_{\ominus r}$ may be difficult 
to compute. To overcome this problem, we use the Hamilton-principle
advocated in \cite{Stoyan} and estimate $\ell(W_{\ominus r})\nu(C)$ by
\[
\sum_{(a,b)\in Y \cap (W_{\ominus r} \times C) } \frac{1}{\lambda(a,b)}.
\]
The result is a ratio-unbiased estimator with the desirable property that
it takes the value one at $r = 0$. Simulations suggest that the Hamilton
principle is also helpful in reducing the sensitivity of the estimator 
with respect to misspecification of the intensity function $\lambda$.
To define a non-parametric estimator for $J_{\rm inhom}^{CD}(r)$, 
we use Theorem \ref{ThmMarkedJfun} and plug in the estimators for the
numerator and denominator discussed above.

\subsection{Hypothesis testing}
\label{SectionTesting}

In Section~\ref{SectionIndependence} we encountered two interaction 
hypotheses: random labelling and independence. Such hypotheses are
complex, depending as they do both on the marginal distribution
of the components of interest and the marking structure or interactions
between them. Nevertheless, it is possible to construct non-parametric
Monte Carlo tests by proper conditioning 
\cite{BesaDigg77,LotwickSilverman,Ripl77} based on a realisation of
$Y$ with locations in some compact window $W \subseteq \R^d$. 

First consider the random labelling hypothesis of Definition~\ref{D:marking}.
Since conditional on the ground process $Z \cap W = \{ z_1, \dots, z_n \}$, the
marks are independent and identically distributed, a Monte Carlo test 
may be based on random permutations of the $n$ observed marks -- in 
effect conditioning on the empirical mark distribution. Another approach
would be to sample the marks according to the mark distribution, but the
latter is typically unknown in practice.

In general, testing for independence is hard.
For hyper-rectangular windows, a Lotwick-Silverman type test 
\cite{LotwickSilverman} can be constructed. Recall that when $Y$ is stationary, 
the key idea to test for independence of $Y_C$ and $Y_D$ 
(where $C, D \in \BB(\M)$ have strictly positive $\nu$-measure)
is to wrap $Y$ onto a torus by identifying opposite sides of $W$, keeping 
$Y_C \cap (W\times \M)$ fixed and translating $Y_D \cap (W\times \M)$ 
randomly over the torus (or vice versa). 
Since the random translations leave the distribution of the $D$-component 
unchanged, they can be used for testing. Note that this approach is 
conditional on the marginal structures of the $C$-marked and $D$-marked 
patterns.

For inhomogeneous marked point processes, randomly translating $Y_D$ 
may change its distribution. To compensate, we must also translate 
the intensity. More specifically, consider the random measure
\begin{equation}
\label{e:measure}
\Xi^Y = \sum_{(z,m) \in Y} \frac{\delta_{(z,m)}}{\lambda(z, m)}
\end{equation}
and let 
\(
\Xi_a^Y 
\)
be its translation over $a\in \R^d$, that is, 
\[
\Xi_a^Y(A) = \Xi^Y(A_{-a}) =
\sum_{(z,m)\in Y} \frac{ \1\{ (z+a, m) \in A \} }{ \lambda(z,m) } =
\sum_{(z,m)\in (Y + a)} 
\frac{ \1\{ (z, m) \in A \} }{ \lambda(z-a,m) }.
\]
Note that $Y$ is translated over $a$ and $\mu(z,m) = \lambda(z-a, m)$ 
is a translation of $\lambda$ over the spatial vector $a$.
Moreover, if $Y$ is intensity-reweighted moment stationary, $\Xi$ is
moment stationary. In other words, for any $a$, $\Xi_a$ has the same 
factorial moment measures as (\ref{e:measure}).

\begin{prop}
\label{LemmaLotwickSilverman}
Let the assumptions of Theorem~\ref{ThmMarkedJfun} be satisfied
and $C, D$ be disjoint Borel mark sets with strictly positive $\nu$-measure.
If $Y_C$ and $Y_D$ are independent and $\Xi$ is stationary, then 
$\widehat D_{\rm inhom}^{CD}$, 
$\widehat F_{\rm inhom}^{D}$ and $\widehat J_{\rm inhom}^{CD}$ can be expressed 
in terms of $(\Xi^{Y_C}, \Xi^{Y_D})$ and their laws are invariant under 
translations of a component.
\end{prop}

\begin{proof}
For $a \in \R^d$,
\begin{align*}
&
\prod_{(z,m)\in (Y+a)}
\left(1-\frac{\bar\lambda_D \1\{ (z,m) \in B(0,r) \times D \}}{
\lambda(z-a, m)}\right)
\\
&=
1 + \sum_{n=1}^{\infty}\frac{(-\bar\lambda_D)^n}{n!} \left[
\sum_{(z_1, m_1), \ldots, (z_n, m_n) \in (Y+a)}^{\neq}
\prod_{i=1}^{n}
\frac{\1\{(z_i, m_i) \in B(0,r) \times D \}}{\lambda(z_i- a, m_i)} \right]
\end{align*}
by the local finiteness of $Y$ and an inclusion-exclusion argument. 
The inner summand is the $n$-th order factorial power measure 
$\xi^{[n]}_{Y_D}$ of
$\Xi^{Y_D}_a$ evaluated at $B(0,r)^n$. Hence (\ref{EstF}) is a function 
of $\Xi^{Y_D}$. Furthermore, as $C$ and $D$ are disjoint, (\ref{EstD})
can be written as
\[
\int_{W_{\ominus r}\times C} \left( 
1 + \sum_{n=1}^{\infty}\frac{(-\bar\lambda_D)^n}{n!} 
\xi^{[n]}_{Y_D}(B(z,r)^n) \right) 
d\Xi^{Y_C}(z,m),
\]
which is well-defined by the local finiteness of $Y$. Finally
$\sum_{(z,m) \in Y \cap (W_{\ominus r} \times C)} \lambda(z,m)^{-1} =
\Xi^{Y_C}(W_{\ominus r}\times C)$ is a function of $\Xi^{Y_C}$ only.
Since $( \Xi^{Y_C}, \Xi^{Y_D}_a ) \stackrel{d}{=} ( \Xi^{Y_C}, \Xi^{Y_D} )$
by the independence of $Y_C$ and $Y_D$ and the invariance under
translation of the law of $\Xi^{Y_D}$, the proof is complete. 
\end{proof}


\section{Application}
\label{SectionApplications}

\subsection{Data description}

In this section, we will apply our statistics to data which are 
presented in \cite{Turner} and available in the {\tt R} package 
{\em spatstat} \cite{BaddeleyTurner}. These data were collected by 
the New Brunswick Department of Natural Resources and cleaned by
Professor Turner. They contain records of wildfires which occurred 
in New Brunswick, Canada, during the years 1987 through 2003; 
records for 1988 are missing. 

More formally, the data $\{ (z_i, m_i, t_i) \}_{i=1}^{n_0}$ consist
of $n_0 = 7108$ recordings of spatial locations $z_i$ of wildfires. 
Attached to each location are two marks: $t_i\in T = \{ 1987, 1989,
\ldots, 2003 \}$ gives the year of occurrence and $m_i$ indicates the 
fuel type. There are four types, of which the dominant one is `forest'
(accounting for some $65\%$ of the fires). The other three types account 
for only about ten percent of the fires each. 

Below, we will quantify interaction in a particular year (here 2000)
using the data in other years to estimate the intensity function. 
We restrict the study area to the rectangular region 
$W = [245.4663, 682.2945] \times [301.0545, 838.6173]$, a subset of New 
Brunswick. Doing so, we obtain the data set $\{(z_i,m_i,t_i)\}_{i=1}^{n}$ 
containing $n=3267$ records of which $147$ occur in the year 2000. 
Since the number of fires occuring in 2000 and fuelled by e.g. `dump' 
is small, we use the mark space $\mathcal{M} = \{\text{forest, other}\}$.
For further details, see \cite{Turner}.


\subsection{Testing for independence}
\label{S:Turner}

To test for independence between the various categories (cf.\ 
Proposition~\ref{LemmaIndependence1}) we use $J_{\rm inhom}^{CD}$ 
in combination with the Lotwick--Silverman approach discussed in 
Section~\ref{SectionTesting}.

\begin{figure}[!htbp]
\begin{center}
\begin{tabular}{cc}
\includegraphics[width=0.45\textwidth]{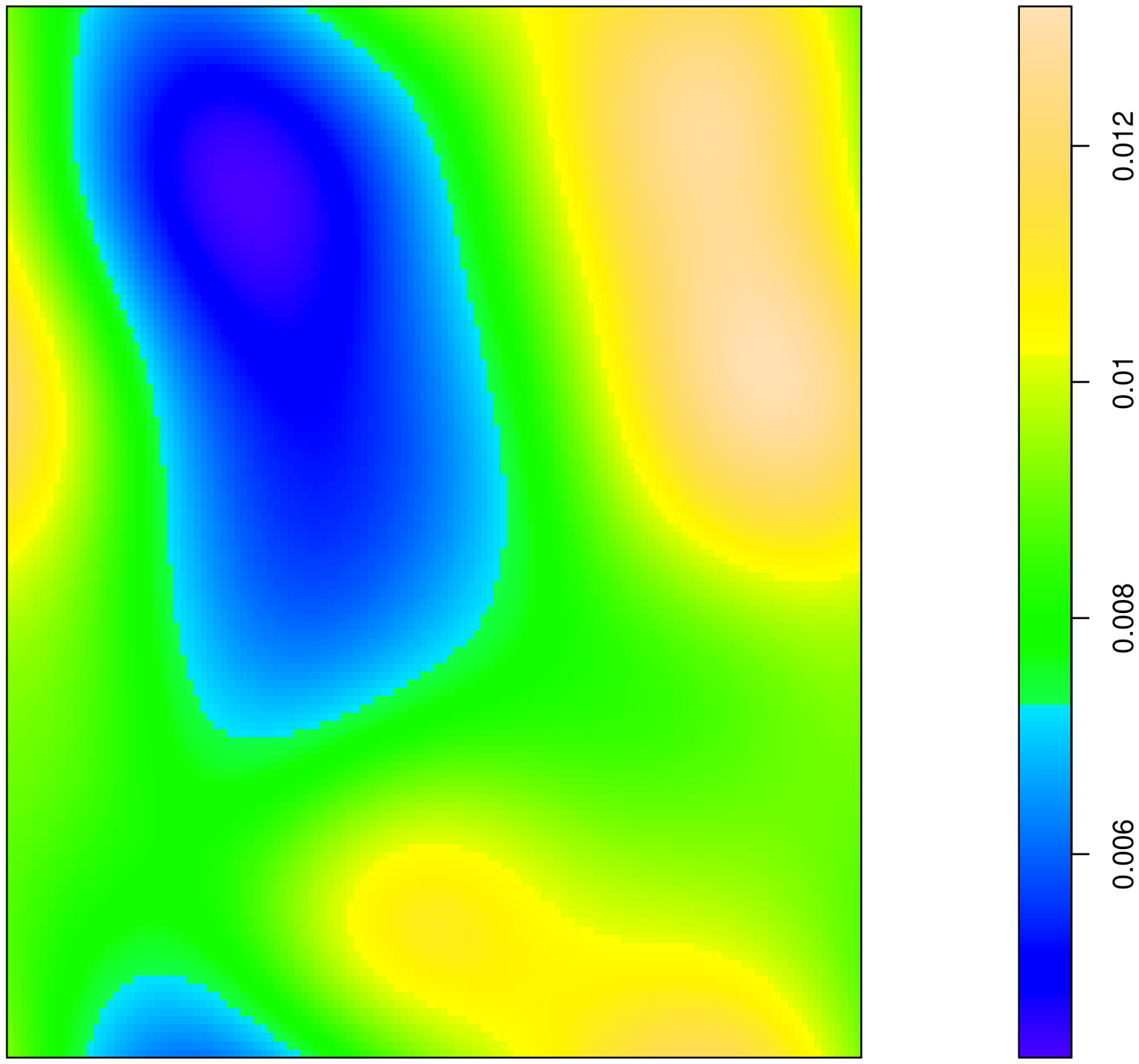}
&
\includegraphics[width=0.45\textwidth]{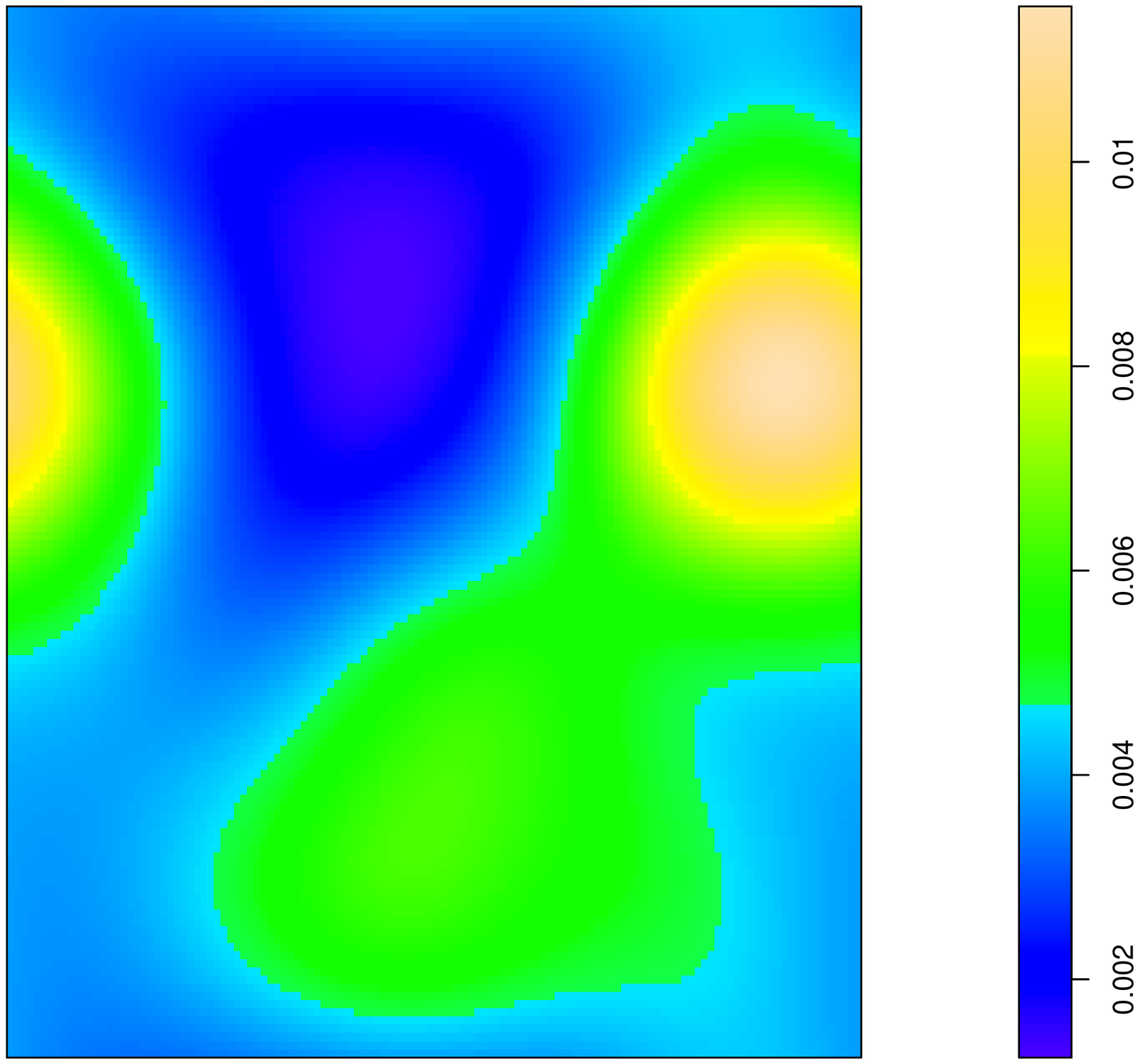}
\end{tabular}
\caption{Gaussian kernel estimator with bandwidth $\sigma=66$ and
torus edge correction calculated over all years except 2000.
Left: forest fuelled fires. Right: other fires.}
\label{F:Intensity}
\end{center}
\end{figure}

Following \cite[p.~205]{Turner}, we assume that for any given year $t\in T$,
the intensity function is of the form 
\begin{equation}
\label{e:scaling}
\lambda^*(z,m,t) = c_t \, \lambda(z,m), \quad 
(z,m) \in W \times \mathcal{M},
\end{equation}
where $c_t>0$ is a year dependent scaling of some overall 
intensity $\lambda(z,m)$. Since the mark set is finite, we take the product
of Lebesgue and counting measure as reference measure so that for fixed mark 
$m$, $\lambda(z, m) = \lambda_m(z)$ is the overall intensity function
of wildfires with fuel type $m$. 

From now on, focus on the year 2000. Since the Lotwick--Silverman approach 
is based on torus translations of one of the component patterns {\em as well
as} its corresponding intensity function, we must use a torus edge correction
for the intensity function. More precisely, we estimate $\lambda_m(z)$
by means of a Gaussian kernel estimator based on all observations with mark 
$m$ that do not fall in the year 2000. Regarding the bandwidth, since we 
consider each mark separately, we use the larger of the bandwidths considered 
by \cite{Turner} for the ground process $Z$, that is, $\sigma = 66$ which 
(approximately) is the square root of the area of New Brunswick multiplied 
by $0.10$.  The results are displayed in Figure~\ref{F:Intensity}.

Recall that the summary statistics discussed in this paper assume that the 
ground process is simple. As the data pattern $\{ ( z_i, m_i, 2000) \}$ for
the year 2000 contains duplicated points, we must discard them, in other words, 
we delete all pairs $(z_i,m_i,t_i)$ and $(z_j,m_j,t_j)$ satisfying $z_i=z_j$ 
and $t_i=t_j=2000$. This results in a marked point pattern with $n_{2000}=124$ 
points (see Figure~\ref{F:Data}). 

\begin{figure}[!hbtp]
\begin{center}
\includegraphics[width=0.45\textwidth]{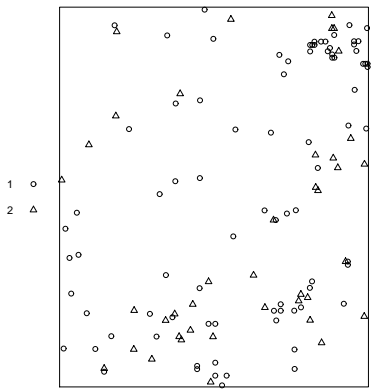}
\caption{Cleaned wildfire data for the year 2000 (\cite{Turner}). 
Type 1 fires are those fuelled by `forest', type 2 fires are fuelled
by other materials.} 
\label{F:Data} 
\end{center}
\end{figure}

To estimate the year dependent constant in (\ref{e:scaling}), we use a 
mass preservation property and equate 
\[
\hat c_{2000} \int_W \sum_{m\in\M} \hat\lambda_m(z) \, dz 
\]
to $n_{2000}$ to obtain $\hat c_{2000} = 124/3120 \approx 0.0397$.
As an aside, an alternative model would be to replace the scaling 
in (\ref{e:scaling}) by a mark dependent one. For $m=\text{`forest'}$, 
this would lead to the value $\hat c_{2000}(m) \approx 0.0414$, which 
does not differ much from $\hat c_{2000}$. 

Set $C=\{\text{forest}\}$ and write $D = \M \setminus C$. To assess
whether the point process $Y_C$ of forest fires occurring in 2000 
is correlated with $Y_D$, we carry out the Lotwick-Silverman scheme 
developed in Section~\ref{SectionTesting} and plot envelopes of
$\widehat{J}_{\rm inhom}^{CD}$ as discussed in Section~\ref{EstimateD}
based on $99$ independent random torus translations. The results are shown in 
Figure~\ref{F:J} and provide graphical evidence for positive correlation 
between the patterns $Y_C$ and $Y_D$.  
This could possibly be interpreted as sparks being transmitted from, 
say, a forest fire to some place further away, where 
the ignition takes place in some other matter, e.g\ a field of grass. 
Furthermore, it may also be an indication that during certain periods 
particular regions are more dry and thus more likely to provide fuel 
for fires. However, since we do not have any specific temporal information 
connected to each point in the data set, besides in which year a fire 
occurs, such conclusions should be treated as speculative. 

\begin{figure}[!htbp]
\begin{center}
\begin{tabular}{cc}
\includegraphics[width=0.45\textwidth]{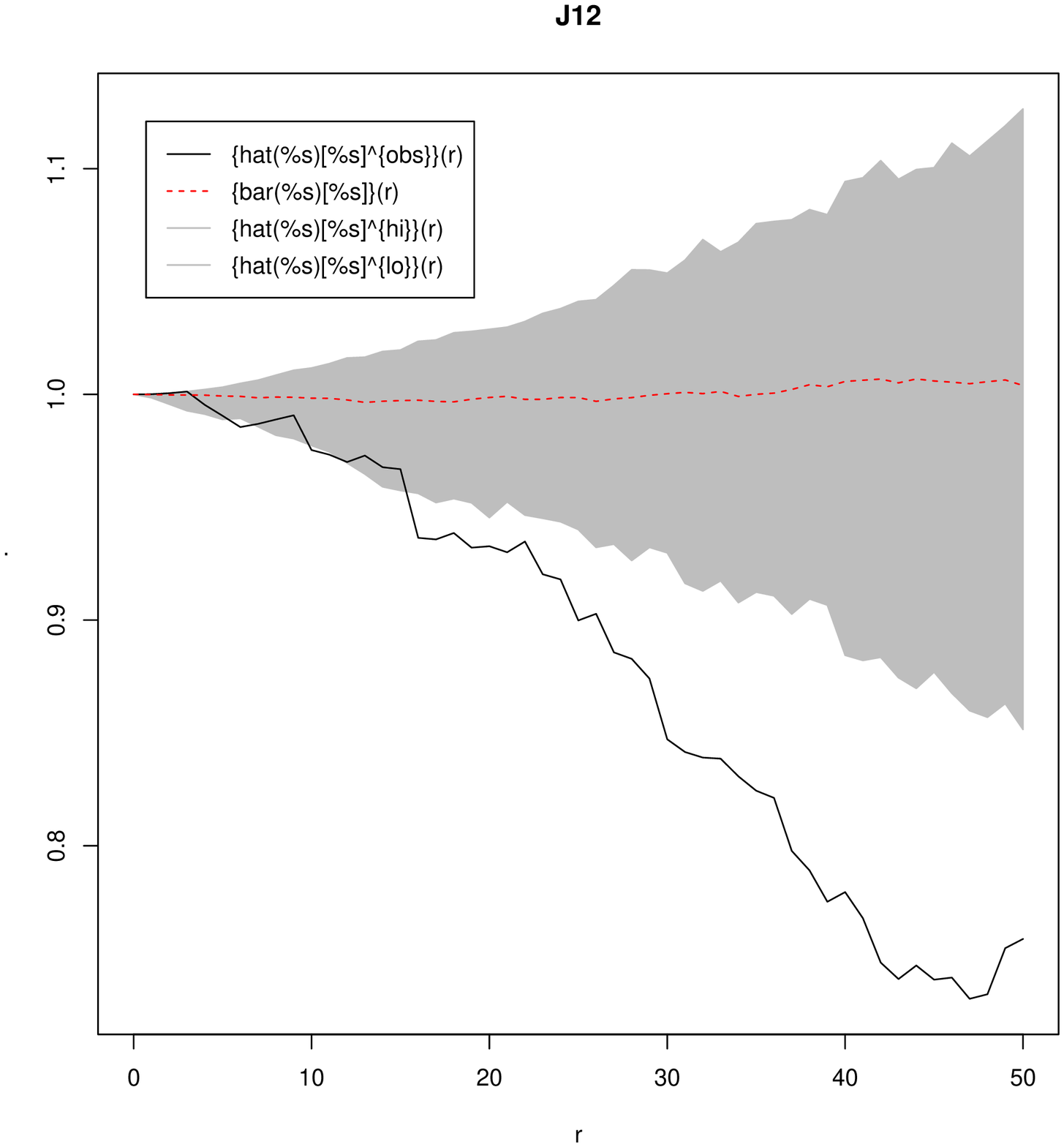}
&
\includegraphics[width=0.45\textwidth]{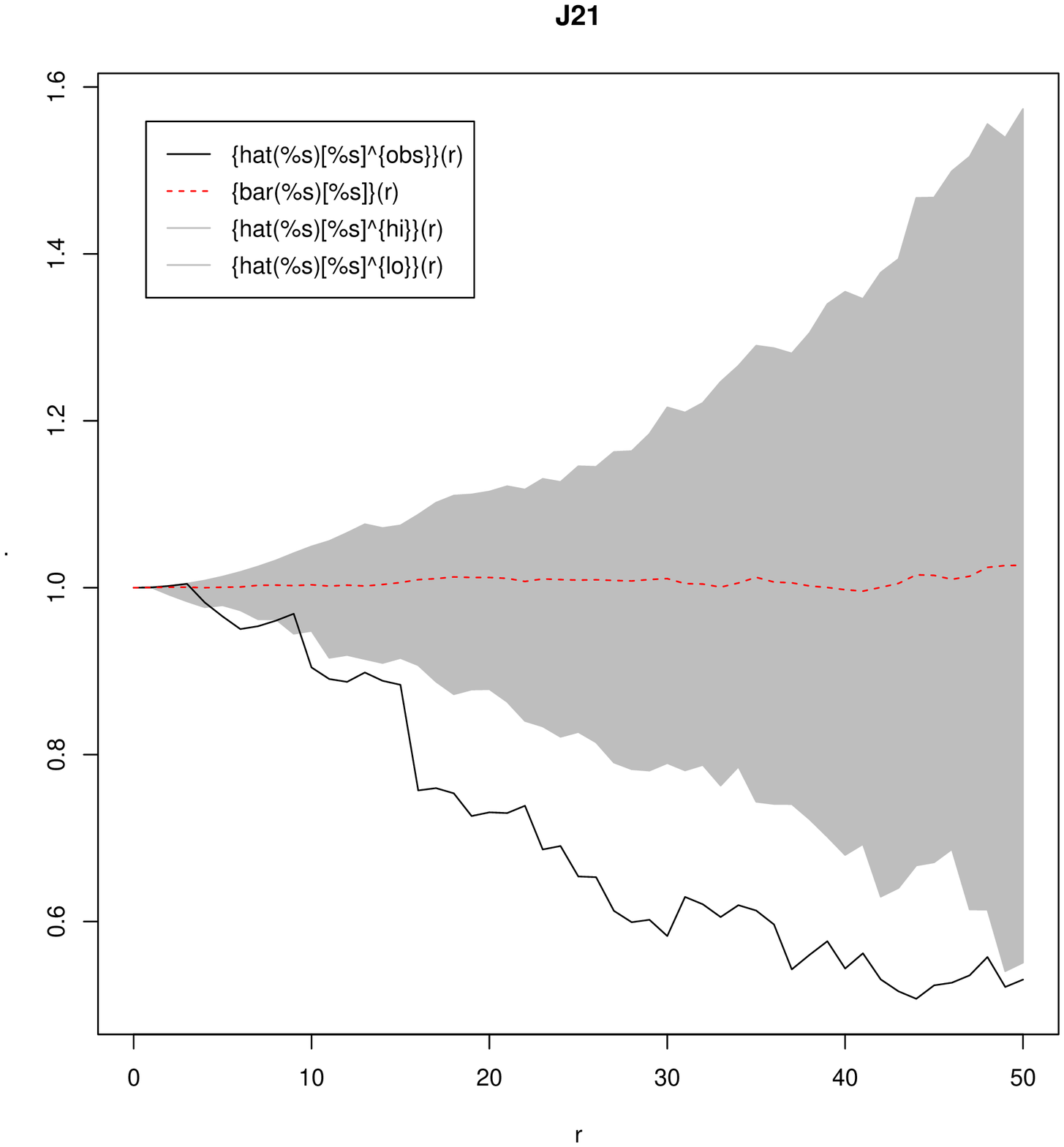}
\end{tabular}
\caption{Estimated inhomogeneous cross J-function for the pattern
displayed in Figure~\ref{F:Data} (black line) with rank $5$-envelopes 
based on $99$ independent translations of the pattern with fuel type 
`forest' and their mean (red line). The plugged in estimated intensity 
functions of the component patterns are proportional to those in 
Figure~\ref{F:Intensity}.
Left: $\widehat{J}_{\rm inhom}^{CD}$-function. 
Right: $\widehat{J}_{\rm inhom}^{DC}$-function.}
\label{F:J}
\end{center}
\end{figure}

\subsection{Random labelling}

Although in our context a random labelling assumption does not seem to be 
appropriate, we perform the analysis for illustrative purposes. 

Recall that under the random labelling assumption, the intensity function
factorises as a product of the ground intensity and the mark distribution.
Therefore, model (\ref{e:scaling}) cannot be used. Instead, for fixed $t$,
let
\[
\lambda^*(z, m, t) = c_t \, \lambda_g(z) 
\]
with respect to the product of the Lebesgue measure and the (empirical) 
mark distribution.  We estimate the ground intensity function $\lambda_g$
by means of a locally edge-corrected Gaussian kernel estimator (see e.g.\ 
\cite{Lieshout12}) with bandwidth $\sigma = 66$ that uses the locations 
of wildfires in all years except for the year 2000. The year dependent 
constant $c_t$ for $t=2000$ is estimated by the mass preservation principle 
as before. 


Appealing to Proposition~\ref{LemmaIndependentlyMarked}, we use $D^{C\M}$ 
with $C= \{ \text{forest} \}$ as test statistic and randomly permute the 
marks of the wildfires in 2000 (see Figure~\ref{F:Data}) $99$ times to 
obtain envelopes. The result displayed in Figure~\ref{F:D} provides no 
evidence for rejecting the random labelling hypothesis.

\begin{figure}[!htbp]
\begin{center}
\includegraphics[width=0.45\textwidth]{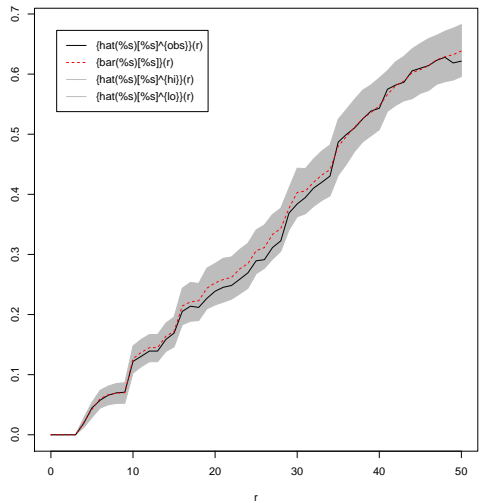}
\caption{Estimated inhomogeneous $D^{C\M}$-function for the pattern
displayed in Figure~\ref{F:Data} (black line) with $C=\{ \text{forest} \}$
with rank $5$-envelopes based on $99$ independent random labellings 
and their mean (red line). The ground intensity function is estimated
by a locally edge-corrected Gaussian kernel estimator 
with bandwidth $\sigma = 66$.}
\label{F:D}
\end{center}
\end{figure}

\section{Summary}
\label{SectionDiscussion}

In this paper we defined cross $D$- and $J$-functions for inhomogeneous 
intensity-reweight\-ed moment stationary marked point processes and 
indicated how they could be used to investigate various independence 
and marking assumptions. In practice, the intensity function tends to 
be unknown and must be estimated. This is not a problem when there are 
independent replicates or pseudo-replication in the form of covariates
\cite{Guan08}. Otherwise, pragmatic model assumptions must be 
made. For example, in Section~\ref{SectionApplications}, we worked under the
assumption that the spatial trend in a given year is proportional to 
the long term trend. When prior information about the data is available, 
a parametric model may also be used.

Finally, it is important to realise that not all point processes on product
spaces are marked point processes, as the ground process need not
be locally finite. An important special class is formed by spatio-temporal
point processes. This class is the focus of a companion paper \cite{CronLies13}
in which we define inhomogeneous $D$- and $J$-functions directly by equipping 
the product space $\R^d \times \R$ with a suitable metric.

\section*{Acknowledgements}

The authors are grateful to R.\ Turner for useful discussions and
access to data. This research was supported by the Netherlands Organisation 
for Scientific Research NWO (613.000.809).

\newpage

\appendix
\appendixpage

\section{Proof of Theorem~\ref{ThmMarkedJfun}}

Consider the function $u_r^{a,D}$ defined in Definition~\ref{DefinitionDfunction}
and let $S_r^D = B(0,r) \times D$.
Then the expansion (\ref{GFMPP}) implies
\beann
G(1-u_r^{a,D}) 
&=& 
1 + \sum_{n=1}^{\infty} \frac{(-1)^n}{n!} 
\int_{\R^d\times\M} \cdots \int_{\R^d\times\M} 
\prod_{i=1}^{n}
\frac{\bar\lambda_D \1\{ (z_i, m_i) \in a + S_r^D \}}{\lambda(z_i, m_i)}
\\
& &\times
\rho^{(n)}((z_1, m_1), \ldots, (z_n, m_n)) 
\prod_{i=1}^{n} dz_i d\nu(m_i) 
\\
&=& 
1 + \sum_{n=1}^{\infty} \frac{(-\bar\lambda_D)^n}{n!} 
\int_{(a+S_r^D)^n} 
\frac
{\rho^{(n)}((z_1, m_1), \ldots, (z_n, m_n))}
{\lambda(z_1, m_1) \cdots \lambda(z_n, m_n)}
\prod_{i=1}^{n} dz_i d\nu(m_i),
\eeann
which is absolutely convergent by assumption and does not depend
on the choice of $a$ by the IRMS-assumption on $Y$. 
Furthermore, for any $a\in\R^d$,
\begin{align*}
&G^{!a}_{C}(1-u_r^{a,D})
=
\frac{1}{\nu(C)} \int_{C} 
\E^{!(a,b)}\left[
\prod_{(z,m)\in Y}
\left(1-\frac{\bar\lambda_D \1\{ (z,m) \in a + S_r^D \}}{\lambda(z, m)}\right)
\right]
d\nu(b)
\\
&=
1+
\frac{1}{\nu(C)} 
\sum_{i=1}^{\infty}\frac{(-\bar\lambda_D)^n}{n!}
\int_{C}
\E^{!(a,b)}\left[
\sum_{(z_1, m_1), \ldots, (z_n, m_n) \in Y}^{\neq}
\prod_{i=1}^{n}
\frac{\1\{(z_i, m_i) \in a + S_r^D \}}{\lambda(z_i, m_i)}
\right]
d\nu(b)
\end{align*}
by the local finiteness of $Y$ and an inclusion-exclusion argument. 

Next, we show that for any bounded $B\in\BB(\R^d)$, 
\beann
&&\int_{B} \left\{ \int_{C}
\E^{!(a,b)}\left[
\sum_{(z_1, m_1), \ldots, (z_n, m_n) \in Y}^{\neq}
\prod_{i=1}^{n}
\frac{\1\{(z_i, m_i)\in a + S_r^D\}}{\lambda(z_i, m_i)}
\right]
d\nu(b) \right \} da
=\\
&=&
\int_{B} \left\{ \int_{C}
\bigg(
\int_{(S_r^D)^n} 
\frac{
\rho^{(n+1)}((0,b), (z_1, m_1), \ldots, (z_n, m_n))
}{\lambda(0,b) \lambda(z_1, m_1) \cdots \lambda(z_n, m_n)}
\prod_{i=1}^{n} dz_i d\nu(m_i)
\bigg)
d\nu(b) \right\} da
\eeann
so that the integrands in between the curly brackets are $\ell$-almost 
everywhere equal and consequently the integrand on the left hand side is 
constant as a function of $a\in\R^d$. In order to do so, define 
\beann
g_{r}^B( (a,b), \varphi) 
= 
\frac{\1\{(a,b) \in B\times C\}}{\lambda(a, b)}
\sum_{(z_1, m_1), \ldots, (z_n, m_n) \in \varphi}^{\neq}
\prod_{i=1}^{n}
\frac{\1\{(z_i, m_i) \in a + S_r^D \}}{\lambda(z_i, m_i)},
\eeann
which is non-negative and measurable.
By rewriting the expression for $g_{r}^B((a,b),Y\setminus\{(a,b)\})$, 
(\ref{CampbellMPP}) and the translation invariance of the $\xi_n$, we obtain
\beann
&&\E\left[\sum_{(a,b)\in Y} g_{r}^B( (a,b), Y\setminus\{(a,b)\})\right]
\\
&=&
\E\left[
\sum_{(a,b), (z_1, m_1), \ldots, (z_n, m_n) \in Y}^{\neq}
\frac{\1\{(a,b)\in B\times C\}}{\lambda(a,b)}
\prod_{i=1}^{n}
\frac{\1\{ (z_i, m_i) \in a + S_r^D \}}{ \lambda(z_i, m_i)}
\right]
\\
&=&
\int_{B} \int_{C}
\bigg(
\int_{S_r^D} \cdots \int_{S_r^D}
\frac{
\rho^{(n+1)}((0,b), (z_1, m_1), \ldots, (z_n, m_n))
}{\lambda(0,b) \lambda(z_1, m_1) \cdots \lambda(z_n, m_n)}
\prod_{i=1}^{n} dz_i d\nu(m_i)
\bigg)
da d\nu(b).
\eeann
At the same time, the reduced Campbell-Mecke formula (\ref{CMthmMPP}) 
implies that
\begin{align*}
& \E\left[\sum_{(a,b) \in Y} g_{r}^B((a,b), Y\setminus\{(a,b)\})\right]
 = 
\\
&
\int_B \int_C
\E^{!(a,b)}\left[
\sum_{(z_1, m_1), \ldots, (z_n, m_n) \in Y}^{\neq}
\prod_{i=1}^{n}
\frac{\1\{(z_i, m_i) \in a + S_r^D \}}{\lambda(z_i, m_i)}
\right]
da d\nu(b)
\end{align*}
and the required equality of the two expressions follows. 
Hence, for $\ell$-almost all $a\in\R^d$, provided the series is 
absolutely convergent,
\begin{align*}
&G^{!a}_{C}(1-u_r^{a,D})
= 1 + \frac{1}{\nu(C)} 
\sum_{n=1}^{\infty} \frac{(-\bar \lambda_D)^n}{n!}
\int_{C} 
\\
&
\bigg(
\int_{S_r^D} \cdots \int_{S_r^D}
\frac{
\rho^{(n+1)}((0,b), (z_1, m_1), \ldots, (z_n, m_n))
}{\lambda(0,b) \lambda(z_1, m_1) \cdots \lambda(z_n, m_n)}
\prod_{i=1}^{n} dz_i d\nu(m_i)
\bigg)
d\nu(b)
\\
&=
1 + \frac{1}{\nu(C)} 
\sum_{n=1}^{\infty} \frac{(-\bar \lambda_D)^n}{n!}
\int_{C} 
\\
&
\bigg(
\int_{S_r^D} \cdots \int_{S_r^D}
\sum_{k=1}^{n+1}
\sum_{E_1, \ldots, E_k}
\prod_{j=1}^{k}
\xi_{|E_j|}(\{(z_i, m_i): i \in E_j\})
\prod_{i=2}^{n+1} dz_i d\nu(m_i)
\bigg)
d\nu(b),
\end{align*}
where $(z_1,m_1) \equiv (0,b)$. 
By splitting this expression into terms based on whether the index sets 
$E_j$ contain the index $1$ (i.e.\ on whether $\xi_{|E_j|}$ includes 
$(z_1,m_1) \equiv (0,b)$), under the convention that $\sum_{k=1}^{0} = 1$, 
we obtain
\begin{align*}
&G^{!a}_{C}(1-u_r^{a,D})
= 1 + 
\frac{1}{\nu(C)} \sum_{n=1}^{\infty} \frac{(-\bar\lambda_D)^n}{n!} 
\sum_{E\in\mathcal{P}_n}
J_{|E|}^{CD}(r)
\sum_{k=1}^{n-|E|}
\sum_{\substack{E_1, \ldots, E_k \neq \emptyset \text{ disjoint}\\ 
\cup_{j=1}^{k} E_j = \{1, \ldots, n \} \setminus E}}
\prod_{j=1}^{k} I_{|E_j|},
\\
&I_{n} = 
\int_{S_r^D} \cdots \int_{S_r^D}
\xi_{n}((z_1, m_1), \ldots, (z_{n}, m_{n}))
\prod_{i=1}^{n} dz_i d\nu(m_i),
\end{align*}
where $J_0^{CD}(r) \equiv \nu(C)$, $|\cdot|$ denotes cardinality and 
$\mathcal{P}_{n}$ the power set of $\{1,\ldots,n\}$. Finally,
by noting that the expansion contains terms of
the form $J_k I_{l_1}^{m_1} \cdots I_{l_n}^{m_n}$ multiplied by a scalar
and basic combinatorial arguments, we conclude that 
\begin{align*}
&G^{!a}_{C}(1-u_r^{a,D})
= \frac{1}{\nu(C)}\left(\nu(C) + 
\sum_{n=1}^{\infty} \frac{(-\bar\lambda_D)^n}{n!} J_n^{CD}(r)\right)
\\
&\times
\Bigg(
1 + \sum_{l=1}^{\infty} \frac{(-\bar\lambda_D)^l}{l!} 
\sum_{k=1}^{l}
\sum_{\substack{E_1, \ldots, E_k \neq \emptyset\text{ disjoint}\\ 
\cup_{j=1}^{k} E_j = \{ 1, \ldots, l \}}}
\prod_{j=1}^{k} I_{|E_j|}
\Bigg)
=
J_{\rm inhom}^{CD}(r) \, G(1-u_r^{0,D}).
\end{align*}
The right hand side is absolutely convergent as a product of 
absolutely convergent terms, therefore so is the series expansion
for $G^{!a}_{C}(1-u_r^{a,D})$. 


\begin{thebibliography}{99}

\bibitem{BaddeleyEtAl}
Baddeley, A. J., M{\o}ller, J. \& Waagepetersen, R. (2000). 
Non- and semi-parametric estimation of interaction in 
inhomogeneous point patterns. 
{\em Statist. Neerlandica} {\bf 54}, 329--350.

\bibitem{BaddeleyTurner}
Baddeley, A. \& Turner, R. (2005).
Spatstat: An {\bf R} package for analyzing spatial point patterns. 
{\em Journal of Statistical Software} {\bf 12}, 1--42.

\bibitem{BesaDigg77}
Besag, J. \& Diggle, P. J. (1977).
Simple Monte Carlo tests for spatial pattern.
{\em Appl. Statist.} {\bf 26}, 327--333.

\bibitem{SKM}
Chiu, S. N., Stoyan, D., Kendall, W. S. \& Mecke, J. (2013).
{\em Stochastic Geometry and its Applications}, 3rd edn,
Wiley, Chichester.

\bibitem{CronLies13}
Cronie, O. \& Lieshout, M. N. M.~van (2013). 
A $J$-function for inhomogeneous spatio-temporal point processes. 
ArXiv, November 2013.



\bibitem{DVJ1}
Daley, D. J. \& Vere--Jones, D. (2003).
{\em An Introduction to the Theory of Point Processes: 
Volume I: Elementary Theory and Methods}, 2nd edn,
Springer, New York.

\bibitem{DVJ2}
Daley, D. J. \& Vere--Jones, D. (2008).
{\em An Introduction to the Theory of Point Processes: 
Volume II: General Theory and Structure}, 2nd edn,
Springer, New York.

\bibitem{Digglebook}
Diggle, P. J. (2014).
{\em Statistical Analysis of Spatial and Spatio-Temporal Point Patterns},
3rd edn,
CRC Press, Boca Raton.

\bibitem{GabrielDiggle}
Gabriel, E. \& Diggle, P. J. (2009).
Second-order analysis of inhomogeneous spatio-temporal point process data. 
{\em Statist. Neerlandica} {\bf 63}, 43--51.

\bibitem{Handbook}
Gelfand, A. E., Diggle, P. J., Fuentes, M. \& Guttorp, P. (2010). 
{\em Handbook of Spatial Statistics},
Taylor \& Francis, Boca Raton. 

\bibitem{Guan08}
Guan, Y. (2008).
On consistent nonparametric intensity estimation for inhomogeneous 
spatial point processes.
{\em J. Amer. Statist. Assoc.} {\bf 103}, 1238--1247.

\bibitem{Halmos}
Halmos, P. R. (1974). 
{\em Measure Theory},
Springer, New York. 

\bibitem{Illian}
Illian, J., Penttinen, A., Stoyan, H. \& Stoyan, D. (2008). 
{\em Statistical Analysis and Modelling of Spatial Point Patterns},
Wiley, Chichester.

\bibitem{MCbook}
Lieshout, M. N. M.~van (2000). 
{\em Markov Point Processes and their Applications},
Imperial College Press/World Scientific, London/Singapore.

\bibitem{MCJfunMPP}
Lieshout, M. N. M.~van (2006). 
A J-function for marked point patterns.
{\em Ann. Inst. Statist. Math.} {\bf 58}, 235--259. 

\bibitem{MCJfun}
Lieshout, M. N. M.~van (2011). 
A $J$-function for inhomogeneous point processes. 
{\em Statist. Neerlandica} {\bf 65}, 183--201. 

\bibitem{Lieshout12}
Lieshout, M. N. M.~van (2012).
On estimation of the intensity function of a point process.
{\em Methodol. Comput. Appl. Probab.} {\bf 14}, 567--578.

\bibitem{MCBaddeley2}
Lieshout, M. N. M.~van \& Baddeley, A. J. (1999). 
Indices of dependence between types in multivariate point patterns. 
{\em Scand. J. Statist.} {\bf 26}, 511--532.

\bibitem{LotwickSilverman}
Lotwick, H. W. \& Silverman, B. W. (1982).
Methods for analysing spatial processes of several types of points.
{\em  J. Roy. Statist. Soc. Ser. B} {\bf 44}, 406--413.

\bibitem{MollerGhorbani}
M\o ller, J. \& Ghorbani, M. (2012). 
Aspects of second-order analysis of structured inhomogeneous spatio-temporal 
point processes.
{\em Statist. Neerlandica} {\bf 66}, 472--491.

\bibitem{MollerBook}
M\o ller, J. \& Waagepetersen, R. P. (2004).
{\em Statistical Inference and Simulation for Spatial Point Processes},
Chapman \& Hall/CRC, Boca Raton.

\bibitem{Penttinen}
Penttinen A. \& Stoyan D. (1989).
Statistical analysis for a class of line segment processes.
{\em Scand. J. Statist.} {\bf 16}, 153--168.

\bibitem{Ripl77}
Ripley, B. D. (1977).
Modelling spatial patterns (with discussion).
{\em J. Roy. Statist. Soc. Ser. B} {\bf 39}, 172--212.

\bibitem{Stoyan}
Stoyan, D. \& Stoyan, H. (2000). 
Improving ratio estimators of second order point process characteristics. 
{\em Scand. J. Statist.} {\bf 27}, 641--656. 

\bibitem{Turner}
Turner, R. (2009).
Point patterns of forest fire locations.
{\em Environ. Ecol. Stat.} {\bf 16}, 197--223.

\bibitem{White}
White, S. D. M. (1979). 
The hierarchy of correlation functions and its relation to other 
measures of galaxy clustering. 
{\em Monthly Notices of the Royal Astronomical Society} {\bf 186}, 145--154.

\end{thebibliography}
\end{document}